\documentclass[conference]{IEEEtran}
\IEEEoverridecommandlockouts
\usepackage{geometry}
\usepackage{amssymb}
\usepackage{amsmath}
\usepackage{caption}
\usepackage{cite}
\usepackage{mathrsfs}
\usepackage[displaymath,mathlines]{lineno}
\usepackage{color}
\usepackage{overpic}
\usepackage{algorithm,algpseudocode}
\usepackage{algorithmicx}
\usepackage{pbox}
\usepackage{multicol}
\usepackage{amsthm}
\usepackage{epstopdf}

\makeatletter
\newcommand{\doublewidetilde}[1]{{%
		\mathpalette\double@widetilde{#1}%
}}
\newcommand{\double@widetilde}[2]{%
	\sbox\z@{$\m@th#1\widetilde{#2}$}%
	\ht\z@=.5\ht\z@
	\widetilde{\box\z@}%
}
\makeatother

\newtheorem{theorem}{Theorem}
\newtheorem{lemma}{Lemma}

\newtheorem{remark}{Remark}

\definecolor{red}{rgb}{0,0,0}
\begin{document}
\title{\huge User Scheduling for Precoded Satellite Systems with Individual Quality of Service Constraints}
\author{\normalsize Trinh~Van~Chien$^\ast$, Eva~Lagunas$^\ast$, Tung~Hai~Ta$^\dagger$, Symeon~Chatzinotas$^\ast$, and Bj\"{o}rn~Ottersten$^\ast$\\
	\IEEEauthorblockA{$^\ast$Interdisciplinary Centre for Security, Reliability and Trust (SnT), University of Luxembourg, Luxembourg \\
	$^\dagger$School of Information and Communication Technology (SoICT), Hanoi University of Science and Technology, Vietnam\\
	Email:\{vanchien.trinh, eva.lagunas, symeon.chatzinotas, and bjorn.ottersten\}@uni.lu,  	tung.tahai@hust.edu.vn \vspace{-0.3cm}}
	\thanks{This work has been partially supported by the Luxembourg National Research Fund (FNR) under the project FlexSAT “Resource Optimization for Next Generation of Flexible SATellite Payloads” (C19/IS/13696663) and the European Space Agency (ESA) funded activity “CGD - Prototype of a Centralized Broadband Gateway for Precoded Multi-beam Networks”. The views of the authors of this paper do not necessarily reflect the views of ESA.}
}

\maketitle

\begin{abstract}
Multibeam high throughput satellite (MB-HTS) systems will play a key role in delivering broadband services to a large number of users with diverse Quality of Service (QoS) requirements. This paper focuses on MB-HTS where the same spectrum is re-used by all user links and, in particular, we propose a novel user scheduling design capable to provide guarantees in terms of individual QoS requirements while maximizing the system throughput. This is achieved by precoding to mitigate mutual interference. The combinatorial optimization structure requires an extremely high cost to obtain the global optimum even with a reduced number of users.  We, therefore, propose a heuristic algorithm yielding a good local solution and tolerable computational complexity, applicable for large-scale networks. Numerical results demonstrate the effectiveness of our proposed  algorithm on scheduling many users with better sum throughput than the other benchmarks. Besides, the QoS requirements for all scheduled users are guaranteed. 
\end{abstract}
\vspace*{-0.2cm}
\IEEEpeerreviewmaketitle
\vspace*{-0.3cm}
\section{Introduction}
\vspace*{-0.2cm}
Multi-beam high throughput satellite  (MB-HTS) systems are known to provide high-speed broadband services to users or areas that cannot be reached or are not sufficiently covered with conventional terrestrial networks \cite{9210567,8746876}. Unlike mono-beam satellites, the received signal strength can be increased thanks to an array fed reflector that results in high beamforming gains and spatially multiplexed communications, following by significant improvements in the instantaneous throughput \cite{7765141}. The multi-spot beams enable an MB-HTS system to offer more service flexibility to satisfy heterogeneous demands from multiple users sharing the same time and frequency resource.  

The performance of MB-HTS systems with aggressive frequency reuse heavily depends on both the precoding design and the user scheduling mechanism, which should be jointly optimized to obtain the globally optimal performance due to the coupled nature as pointed in \cite{vazquez2016precoding}. Unfortunately, the joint optimization is extraordinarily challenging for real systems since the precoding coefficients are chosen based on the channel state information (CSI) of the scheduled users; and the scheduled users' performance is dependent on the precoding design. De facto, a system performance close to the optimal can be attained when users with semi-orthogonal channel vectors are selected \cite{Yoo2006a,8401547}. By fixing the precoding technique, most of the previous works have focused on the user scheduling designs for a single time slot by estimating the orthogonality between the channel vectors using, for example, the cosine similarity metric \cite{7091022} or the semi-orthogonality projection \cite{Yoo2006a}. However, the user scheduling over multiple time slots, i.e., block scheduling design, will be different and more challenging to maintain the QoSs of scheduled users. To the best of authors' knowledge, it is the first time that MB-HTS block scheduling with individual QoS constraints has been investigated.

This paper explores the benefits of block-based user scheduling in enhancing the system throughput, whilst maintaining the QoS requirements in MB-HTS systems with full frequency reuse. Our main contributions are listed as follows: $i)$ We formulate a novel user scheduling problem spanning different time slots that maximizes the sum throughput for an observed window time and the user-specific QoS constraints. Determining the optimal solution to this combinatorial problem requires an exhaustive search of the parameter space. This is not to be preferable due to the exponential increase of the potential scheduling solutions when many users are available in the coverage area. $ii)$ We, therefore, propose a heuristic algorithm yielding a local solution in polynomial time. We also theoretically provide the convergence analysis and the computational complexity order. $iii)$ The proposed scheduling algorithm is evaluated via numerical simulations and it outperforms the other benchmarks in both the sum and per-user throughput. The users' QoS requirements formulated with specific per-user data demands are shown to be satisfied.  

\textit{Notation}: The upper and lower bold letters denote the matrices and vectors, respectively. The superscripts $(\cdot)^H$ and $(\cdot)^T$ are the Hermitian and regular transposes. The Euclidean norm is $\| \cdot \|$, $\mathrm{tr}(\cdot)$ is the trace of a matrix, and $\mathcal{CN}(\cdot, \cdot)$ is the circularly symmetric Gaussian distribution. The expectation of a random variable is $\mathbb{E}\{\cdot\}$. The union of sets is  $\cup$ and $\subseteq$ denotes the subset operator. Finally, the cardinality of set $\mathcal{A}$ is denoted as $|\mathcal{A}|$ and $\mathcal{O}(\cdot)$ is the big-$\mathcal{O}$ notation. 

\vspace*{-0.3cm}
\section{System Model \& Performance Analysis}
\vspace*{-0.2cm}
This section introduces a unicast multi-beam satellite system model in which a single user per beam is scheduled at each time instance. The aggregated and instantaneous downlink throughput for every scheduled user is then presented under the considered scheduling framework.
\vspace*{-0.2cm}
\subsection{System Model}
\vspace*{-0.2cm}
We  consider  the  downlink  of  a geostationary (GEO)  broadband  MB-HTS system that aggressively reuses the user link frequency. Precoding is assumed to be implemented in order to mitigate the co-channel interference. The satellite is assumed to generate $M$ partially overlapping beam clusters as illustrated in Fig.~\ref{FigSysModel}. For simplicity, the number of overlapping beam clusters is equal to the number of antennas at the satellite. There are $N$ single-antenna users available with $N \gg M$ in an observed window time that comprises $T$ time slots. We assume that the system operates in a unicast mode in which in which at most $M$ users can be scheduled per time-slot (black users in Fig.~\ref{FigSysModel}). By acquiring the fixed-satellite service \cite{guidolin2015study}, user locations are geographically fixed, but the transmit data signals are independently distributed and mutually exclusive. Let us denote $\mathcal{K}(t)$ the scheduled-user set at the $t-$th time slot, which satisfies
$\mathcal{K}(t) \subseteq \{1, \ldots, N \}$  and $|\mathcal{K}(t)| \leq M$.
We assume that in the observed window time, the propagation channels are static, which is in general valid for GEO satellite communications and reasonable window lengths. Specifically, if the channel between user~$k$ and the satellite is $\mathbf{h}_k \in \mathbb{C}^M$, then we can denote the channel matrix at the $t-$th time slot as
$\mathbf{H}(t) = \big[\mathbf{h}_{\pi_1}, \ldots,\mathbf{h}_{\pi_{|\mathcal{K}(t)|}}  \big] \in \mathbb{C}^{M \times |\mathcal{K}(t)|}$ with $\pi_1, \ldots,\pi_{|\mathcal{K}(t) |} $  being the user indices in $\mathcal{K}(t)$. Subsequently, the size of channel matrix depends on the cardinality $|\mathcal{K}(t)|$. From practical aspects, $\mathbf{H}(t)$ is formulated as
\begin{equation}
\mathbf{H}(t) =   \mathbf{B}(t) \pmb{\Phi}(t),
\end{equation}
where $\mathbf{B}(t) \in \mathbb{R}_{+}^{M \times |\mathcal{K}(t)|}$ represents the different influences in satellite communications comprising the received antenna gain, thermal noise, path loss, and satellite antenna radiation pattern with the $(m,k)-$th element defined as
\begin{equation}
b_{mk} = \frac{\lambda \sqrt{\widehat{G}_{Rk} G_{mk}}}{4 \pi d_{mk}}, m=1, \ldots M, k = 1, \ldots, |\mathcal{K}(t)|,
\end{equation}
where $\lambda$ is the wavelength of a plane wave; $d_{mk}$ is the distance between the $m-$th satellite antenna and user~$k$. It is safe to assume $d_{1k} = \ldots = d_{Mk}, \forall k,$ for a GEO satellite system because of long propagation distance.  The receiver antenna gain is denoted as $\widehat{G}_{Rk}$, which mainly depends on the receiving antenna aperture, whilst $G_{mk}$ is the gain defined by the satellite radiation pattern and user location. The diagonal matrix $\pmb{\Phi}(t) \in \mathbb{C}^{|\mathcal{K}(t)| \times |\mathcal{K}(t)|}$ expresses the signal phase rotations with the  $(k,k)-$th diagonal element $\phi_{kk} = e^{i \psi_k}, \forall k = 1, \ldots, |\mathcal{K}(t)|,$ and $\psi_k$ identically and independently distributed by the uniform distribution.
\begin{figure}[t]
	\centering
	\includegraphics[trim=3.9cm 0.6cm 3.5cm 3.2cm, clip=true, width=2.8in]{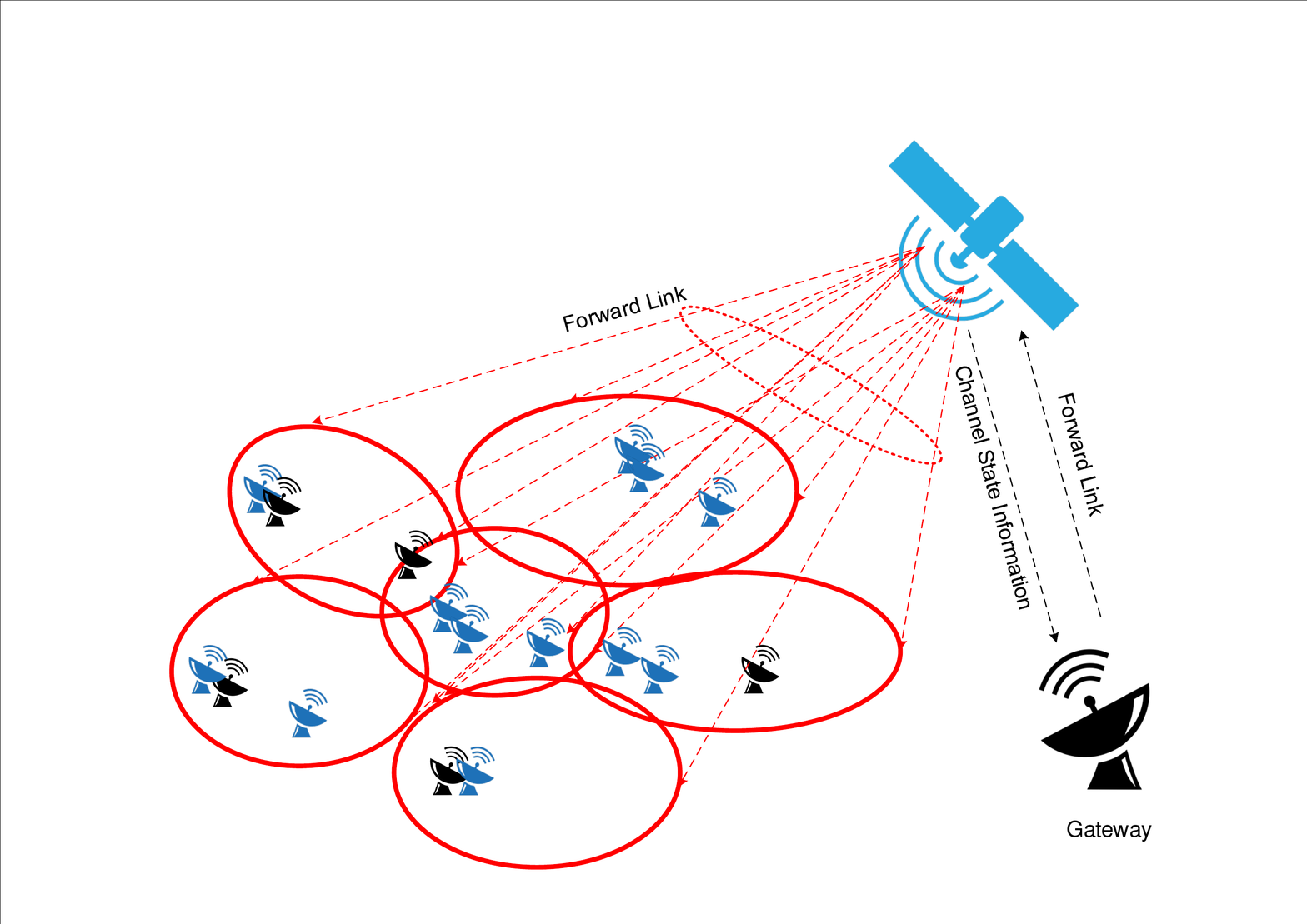} \vspace*{-0.15cm}
	\caption{The MB-HTS system model with one GEO satellite serving many users in an observed window time.}
	\label{FigSysModel}
	\vspace*{-0.5cm}
\end{figure}
\vspace{-0.2cm}
\subsection{Downlink Data Transmission}
\vspace{-0.2cm}
At the $t-$th time slot, the satellite is simultaneously transmitting data signals to the scheduled users. In detail, $s_k(t)$ is the modulated data symbol for scheduled user~$k$ with $ \mathbb{E}\{|s_k(t)|^2\} = 1$. 
The received signal at scheduled user $k$ with $k \in \mathcal{K}(t)$, denoted by $y_k(t) \in \mathbb{C}$, is thus formulated as 
\begin{equation} \label{eq:ReceivedSig}
	y_k(t) = \sum\nolimits_{k' \in \mathcal{K}(t)} \sqrt{p_{k'}} \mathbf{h}_{k}^H \mathbf{w}_{k'}(t) s_{k'}(t) + n_k(t),
\end{equation}
where $\mathbf{w}_k(t)$ is the precoding vector used for scheduled user~$k$ with $\| \mathbf{w} (t) \| = 1$ and $p_k$ is data power allocated to this user; $n_k(t)$ is additive noise with $n_k(t) \sim \mathcal{CN}(0,\sigma^2)$ and $\sigma^2$ being the noise variance. Although the channels are static in the observed window time, the precoding vectors vary upon time slots due to the user scheduling. Conditioned on the precoding vectors, the limited power budget at the satellite can be expressed as
\begin{equation}
	\sum_{k' \in \mathcal{K}(t)} p_{k'} \mathbb{E} \{ \| \mathbf{w}_{k'} (t) s_{k'}(t) \|^2 \} = 	\sum_{k' \in \mathcal{K}(t)} p_{k'} \leq P_{\max},
\end{equation}
where $P_{\max}$ is the maximum transmit power that the satellite can spend  for data symbols at the $t-$th time slot. In order to compute the instantaneous throughput of scheduled user~$k$, we recast the received signal \eqref{eq:ReceivedSig} into an equivalent form as
\begin{equation} \label{eq:ReceivedSigv1}
	\begin{split}
	 y_k(t) =& \sqrt{p_{k}} \mathbf{h}_{k}^H \mathbf{w}_{k}(t) s_{k}(t) +  \sum_{k' \in \mathcal{K}(t) \setminus \{ k\} } \sqrt{p_{k'}} \mathbf{h}_{k}^H \\
	 &\times \mathbf{w}_{k'}(t)  s_{k'}(t) + n_k(t),
	\end{split}
\end{equation}
where the first part contains the desired signal, while the second part is mutual interference from the other scheduled users at the $t-$th time slot. From \eqref{eq:ReceivedSigv1}, the aggregated and per-time-slot throughput of scheduled user~$k$ is given in Lemma~\ref{lemma:ChannelCapacity}.
\begin{lemma} \label{lemma:ChannelCapacity}
Assuming that user~$k$ is scheduled only in the $T_k$ time slots, $1 \leq T_k \leq T$, its aggregated throughput is
\begin{equation} \label{eq:Rk}
	R_k\left( \{ \mathcal{K}(t) \} \right) = \sum_{t=1}^{T_k} R_k(\mathcal{K}(t) ),\mbox{[Mbps]},
\end{equation}
where $R_k(\mathcal{K}(t) )$ is the instantaneous throughput at the $t-$th time slot, $1\leq t \leq T_k$, which is computed as 
\begin{equation}\label{eq:Capacityk}
R_k(\mathcal{K}(t) ) = B \log_2 \left( 1 + \mathrm{SINR}_k(\mathcal{K}(t) ) \right), \mbox{[Mbps]},
\end{equation}
where $B$~[MHz] is the system bandwidth and the signal-to-interference-and-noise ratio is
\begin{equation}
\mathrm{SINR}_k(\mathcal{K}(t) )  = \frac{ p_{k} | \mathbf{h}_{k}^H \mathbf{w}_{k}(t) |^2 }{ \sum_{k' \in \mathcal{K}(t) \setminus \{ k\} } p_{k'} |\mathbf{h}_{k}^H \mathbf{w}_{k'}(t)|^2 +\sigma^2}.
\end{equation}
\end{lemma}
\begin{proof}
The instantaneous throughput of scheduled user~$k$ at each time slot is computed as \eqref{eq:Capacityk} by exploiting the Shannon channel capacity under perfect channel state information and known mutual interference. The aggregated throughput is further accumulated over all the $T_k$ time slots as in  \eqref{eq:Rk}.
\end{proof}
For a given transmit power coefficients, the instantaneous throughput in \eqref{eq:Capacityk} is a function of $\mathcal{K}(t)$, while the aggregated throughput depends on all the scheduled users in the $T_k$ time slots. It is noteworthy that the throughput in Lemma~\ref{lemma:ChannelCapacity} can be applied for arbitrary channel models and precoding techniques. This paper exploits linear precoding processing because it has a lower cost than the optimal. More specifically, we deploy the regularized zero forcing (RZF) precoding matrix $\mathbf{W}(t) \in \mathbb{C}^{M \times |\mathcal{K}(t)|}$, which is
\begin{equation} \label{eq:MMSEPrecodMatrix}
\mathbf{W}(t) = \frac{1}{\sqrt{\gamma(t)}} \mathbf{H}(t) \Big(\mathbf{H}(t)^H \mathbf{H}(t) + \frac{\sigma^2}{P_{\max}} \mathbf{I}_{ |\mathcal{K}(t)|} \Big)^{-1},
\end{equation}
where $\mathbf{I}_{ |\mathcal{K}(t)|}$ is the identity matrix of size $|\mathcal{K}(t)| \times |\mathcal{K}(t)|$ and the normalized power constant $\gamma(t)$ is defined as
\begin{equation}
\gamma(t) = \mathrm{tr} \Big( \mathbf{H}(t) \Big(\mathbf{H}^H(t) \mathbf{H}(t) + \frac{\sigma^2}{P_{\max}} \mathbf{I}_{|\mathcal{K}(t)|} \Big)^{-2} \mathbf{H}^H(t) \Big).
\end{equation}
We should notice that each precoding matrix in \eqref{eq:MMSEPrecodMatrix} is a function of the scheduled users at the $t-$th time slot, thus it verifies the high importance of a proper set $\mathcal{K}(t)$ in boosting the throughput. By counting for the arithmetic operations with the high cost such as complex multiplications and divisions \cite{van2019power}, the computational complexity order to construct a RZF precoding matrix is presented in Lemma~\ref{lemma:ComplexityWtv2}. 
\begin{lemma} \label{lemma:ComplexityWtv2}
 The precoding matrix $\mathbf{W}(t)$ is constructed by the computational complexity in the order of $\mathcal{O}\left( \frac{1}{2}M^2 |\mathcal{K}(t)| \right)$ as a consequence of the channel matrix $\mathbf{H}(t)$ depending on the scheduled-user set $\mathcal{K}(t)$.
\end{lemma}
\begin{proof}
By applying \cite[Lemma~B.1]{Bjornson2017bo} to the channel matrix $\mathbf{H}(t)$, the product $\mathbf{H}^H(t)\mathbf{H}(t)$ requires $\frac{1}{2}|\mathcal{K}(t)|(|\mathcal{K}(t)|+1)M$ complex multiplications thanks to the Hermitian symmetry. Let us introduce a new matrix $\mathbf{G}(t) = \mathbf{H}(t)^H \mathbf{H}(t) + \frac{1}{P_{\max}} \mathbf{I}_{|\mathcal{K}(t)|}$, then attaining $\mathbf{G}(t)$ needs $\left( \frac{1}{2}M(|\mathcal{K}(t)|+1) + 1\right)|\mathcal{K}(t)|$ complex multiplications. According to \cite[Lemma~B.2]{Bjornson2017bo}, the inverse matrix $\mathbf{H}(t)\mathbf{G}^{-1}(t)$ can be computed efficiently by utilizing the Cholesky decomposition that includes $\frac{|\mathcal{K}(t)|^3 - |\mathcal{K}(t)|}{3} + |\mathcal{K}(t)|^2 M + |\mathcal{K}(t)|$ complex multiplications and divisions. Furthermore, we need the $\frac{1}{2}(M^2 +M)|\mathcal{K}(t)| + 2$ complex multiplications, division, and square root to obtain $\gamma(t)$. Thus, the number of the arithmetic operations to obtain the RZF precoding matrix $\mathbf{W}(t)$ is obtained by adding all the cost. 
Due to the fact $|\mathcal{K}(t)| \leq M$, we can ignore the terms with low degree in the obtained posynomial and hence the computational complexity order is shown as in the lemma.
\end{proof}
The key achievement from Lemma~\ref{lemma:ComplexityWtv2} is to point out the computational complexity of the RZF precoding matrix construction directly proportional to the total elements in the scheduled-user set $\mathcal{K}(t)$ for a given number of satellite beams. We later utilize Lemma~\ref{lemma:ComplexityWtv2} to evaluate the complexity order of the proposed algorithm to the user scheduling problem. 

From \eqref{eq:MMSEPrecodMatrix}, the precoding vector dedicated to scheduled user~$k$  at each time slot, i.e., $\mathbf{w}_k(t)$, is selected as the $k-$th column of matrix $\mathbf{W}_k(t)$. By exploiting a similar methodology as what has done for Lemma~\ref{lemma:ComplexityWtv2}, it is straightforward to manifest that RZF precoding has the higher computational complexity than other linear signal processing techniques such as maximum ratio or zero forcing. Nonetheless, this precoding selection provides better throughput than the others and avoiding an ill-posed inverse appearing when the channels are highly correlated leading to rank deficiency.
\vspace*{-0.2cm}
\section{Sum Throughput Optimization}
\vspace*{-0.2cm}
By considering the user scheduling in an observed window time, a sum throughput optimization problem with the QoS requirements is first formulated. Because of the inherent non-convexity, a heuristic algorithm is then proposed to obtain a local solution in polynomial time.
\vspace*{-0.2cm}
\subsection{Problem Formulation}
\vspace*{-0.1cm}
Our objective function in this paper is the total sum throughput of all the scheduled users in the considered window time and the individual QoS requirements of scheduled users are constraints. Mathematically, the optimization problem, which we would like to solve, is formulated as
\vspace*{-0.1cm}
\begin{subequations} \label{Problemv1}
	\begin{alignat}{2}
		& \underset{\{ \mathcal{K}(t) \} }{\mathrm{maximize}}
		& & \, \sum\nolimits_{t=1}^T \sum\nolimits_{k \in \mathcal{K}(t) } R_k(\mathcal{K}(t)) \\
		& \mbox{subject to}
		& &   \, R_k(\mathcal{K}(t)) \geq \frac{\xi_k}{T_k}, \forall k \in \mathcal{K}(t), \forall t, \label{eq:QoSconst}\\
		&&& \, \mathcal{K} (t) \in \{ 1, \ldots, N\}, \forall t, \label{eq:Kt1}\\
		&&& \, | \mathcal{K} (t)| \leq M, \forall t, \label{eq:Kt2} \\
		&&& \,  \bigcup \mathcal{K}(t) \subseteq \{1, \ldots,N \} \label{eq:Kt3},
	\end{alignat}
\end{subequations}
where $T_k$ is the number of time slot that spends on scheduled user~$k$ to fulfill the QoS requirement, denoted by $\xi_k$ as in \eqref{eq:QoSconst}. As $T$ is sufficiently large, the long-term QoS satisfaction of user~$k$ is defined as $R_k(\{ \mathcal{K}(t) \}) \geq \xi_k T_k$, which is spontaneously fulfilled when all the per-time-slot constraints in \eqref{eq:QoSconst} hold. Furthermore, \eqref{eq:Kt1}--\eqref{eq:Kt3} show the conditions on all the scheduled-user sets $\mathcal{K} (t), \forall t$. Specifically, \eqref{eq:Kt1} implies that every $\mathcal{K}(t)$ is a subset of the available-user set, say $\{1,\ldots,N\}$, whilst \eqref{eq:Kt2} implies that the number of scheduled users may be less than the available beams to maximize the sum throughput in the entire network and therefore demonstrating the flexibility of our optimization problem. The union of all the scheduled-user sets $\mathcal{K}(t), \forall t,$ over the observed window time is a subset of the available-user set in general. From the system viewpoint, some users may be ignored from service due to, for example, bad channel conditions and/nor too high QoS requirements such that they are not be served with a limited transmit power level.

We stress that problem~\eqref{Problemv1} is non-convex as a consequence of the discrete feasible domain and the non-convex objective function. Particularly, the discrete feasible domain makes \eqref{Problemv1} a combinatorial problem, where the global optimum can only be obtained for a small scale network  with few users and small number of beams since an exhaustive search of the parameter space is required. Nevertheless, the exhaustive search has the computational complexity scaling up exponentially with the number of available users. For instance, with $M=7, N=100$, and only one time slot is considered for the sake of simplicity, the optimal solution is obtained by searching over $\sum_{k=1}^M \frac{N!}{k!(N-k)!} \approx 1.7 \times 10^{10}$ different combinations, which is prohibitively large. An exhaustive search is, therefore, not preferable for large-scale networks with many users as the main consideration in this paper. For now, we  differentiate our user scheduling optimization problem from the related works as shown in Remark~\ref{Remark1}.
\vspace*{-0.2cm}
\begin{remark} \label{Remark1}
Problem~\eqref{Problemv1} is a generalized version of the previous works  \cite{Yoo2006a,Honnaiah2020} and references therein since the $N$ users are scheduled over different time slots and since we also take the QoS requirements into account. In other words, problem~\eqref{Problemv1} ensures the scheduled users always satisfied their throughput demands. Furthermore, an effective RZF precoding matrix constructed from a good scheduling scheme not only reduces mutual interference but also ameliorates the received signal strength that boosts the system performance. With a limited window time and the correlation among propagation channels, the number of scheduled users might be less than the total available users to maximize the network throughput.
\vspace*{-0.2cm}
\end{remark}
\vspace*{-0.1cm}
\subsection{Proposed Heuristic Algorithm}
\vspace*{-0.1cm}
Motivated by large-scale networks with many users simultaneously requesting to admit the system, we propose a heuristic algorithm that obtains a good local solution in polynomial time with tolerable computational complexity. Algorithm~\ref{Algorithm1} demonstrates the proposal with the double loops: The outer loop indicates the evolution of time slots and the inner loop is for the growth of the scheduled users per time slot.

 At the initial stage, let us denote $\mathcal{N}(0) \leftarrow \{1, \ldots, N\}$ the set of available users with the corresponding channels $\mathbf{h}_1, \ldots, \mathbf{h}_{N}$. Moreover, the scheduled user set $\mathcal{K}(0)$ is initially setup as an empty set. The proposed heuristic algorithm begins with sorting the channel gains in a descending order as
\begin{equation} \label{eq:ChannelOrder}
\| \mathbf{h}_{\pi_1}  \|^2 \geq \| \mathbf{h}_{\pi_2}  \|^2 \geq \ldots \geq \| \mathbf{h}_{\pi_N}    \|^2,
\end{equation}
where $\{\pi_1, \ldots,  \pi_N \}$ is a permutation of the user indices for which \eqref{eq:ChannelOrder} holds. Then, we set the outer iteration index $t=1$ and the available- and scheduled-user sets are updated as
\begin{equation} \label{eq:1stUpdate}
\mathcal{N}(1) \leftarrow \mathcal{N}(0) \setminus \{ \pi_1 \} \mbox{ and } \mathcal{K}(1) \leftarrow \mathcal{K}(0) \cup \{ \pi_1 \}.
\end{equation}
At the $t$-th outer iteration ($1\leq t \leq T$), if the number of scheduled users from the previous time slot, which have not been satisfied their QoS requirements yet, is less than the number of beams, i.e., $|\mathcal{K}(t-1)| < M$, there is room for scheduling new users to join the system in case of all the constraints of problem~\eqref{Problemv1} satisfied. For such, an inner loop is implemented to testify whether or not at most the $M- |\mathcal{K}(t-1)| +1$~potential users can be scheduled. The following optimization problem is therefore considered at the $m-$th inner iteration ($|\mathcal{K}(t-1)| \leq m \leq M$):
\begin{equation} \label{Prob:Ratev1}
k_m^{t,\ast} = \underset{k \in \mathcal{N}(t)}{\mathrm{argmax}}  \sum\nolimits_{k' \in \widetilde{\mathcal{K}}_m(t) } R_{k'} (\widetilde{\mathcal{K}}_m(t)),
\end{equation}
where each set $\widetilde{\mathcal{K}}_m(t)$ is related to one user $k \in \mathcal{N}(t)$, which is defined as
\begin{equation} \label{eq:Ktilde}
\widetilde{\mathcal{K}}_m(t) \leftarrow  \begin{cases}
	\mathcal{K}_{m-1}(t) \cup \{k\}, & \mbox{if } m =  |\mathcal{K}(t-1)| +1, \ldots, M,\\
 \mathcal{K}(t-1) \cup \{k\}, & \mbox{if }  m =  |\mathcal{K}(t-1)|.
\end{cases}
\end{equation}
In \eqref{eq:Ktilde},  $\mathcal{K}_{m-1}(t) $ is the scheduled-user set at the $(m-1)-$th inner iteration with $\mathcal{K}_m (t) = \mathcal{K}(t-1)$ when $m=|\mathcal{K}(t-1)|$. Problem~\eqref{Prob:Ratev1} aims at maximizing the total sum throughput at a particular time slot only.\footnote{The solution to problem~\eqref{Prob:Ratev1} is not unique in general. Alternatively, there may be more than one user with the same total sum throughput, but we can select one of them for further processing.} Hence, the solution to problem \eqref{Prob:Ratev1} does not guarantee a monotonic increasing property, which is in need to have a good local solution to the original problem~\eqref{Problemv1}. As foreseen from a multi-user system, user~$k_m^{t,\ast}$ causes more mutual interference to other users in the set $\widetilde{\mathcal{K}}_m(t)$ that may lead to their throughput no longer satisfy the QoS requirements. In order to get rid of this issue, we suggest a mechanism to further testify whether or not user~$k_m^{t,\ast}$ becomes a scheduled user as in Theorem~\ref{Theorem:SelectedUser}.
\vspace{-0.1cm}
\begin{theorem} \label{Theorem:SelectedUser}
 User~$k_m^{t,\ast}$ becomes a scheduled user if the following conditions satisfy
\begin{align}
  \sum_{k' \in \widetilde{\mathcal{K}}_m^\ast (t) } R_{k'} ( \widetilde{\mathcal{K}}_m^\ast (t) ) &\geq \sum_{k' \in \mathcal{K}_{m-1}(t) } R_{k'} (\mathcal{K}_{m-1}(t)), \label{eq:Ser1}\\
 R_{k'}(t) &\geq \xi_{k'}/T_{k'}, \forall k' \in \widetilde{\mathcal{K}}_m^{\ast}(t), \label{eq:Ser2}
\end{align}
where $\widetilde{\mathcal{K}}_m^{\ast}(t)$ is formulated as in \eqref{eq:Ktilde}, but for user~$k_m^{t,\ast}$. The condition~\eqref{eq:Ser1} guarantees the objective function of problem~\eqref{Problemv1} to be non-decreasing along iterations until reaching a fixed point, while all users admitted to the network satisfy their QoS requirements by the condition~\eqref{eq:Ser2}.
\vspace{-0.15cm}
\end{theorem}
\begin{proof}
We first prove that at the $t-th$ outer iteration, the objective function of problem~\eqref{Problemv1} is non-decreasing along inner iterations. Let us introduce $\alpha_m^t, m \in \{|\mathcal{K}(t-1)|, \ldots, M\},$ as
\begin{equation}
\alpha_m^t = \sum\nolimits_{k' \in \widetilde{\mathcal{K}}_m ^{\ast} (t) } R_{k'} (\widetilde{\mathcal{K}}_m ^{\ast} (t) ),
\end{equation}
then by exploiting \eqref{eq:Ser1} the following series of inequality holds
\begin{equation}
\alpha_M^t \geq \alpha_{M-1}^t \geq \ldots \geq \alpha_{|\mathcal{K}(t-1)|}^t,
\end{equation}
which demonstrates the non-decreasing property of the sum throughput in every time slot. Due to the non-negative property of the instantaneous channel capacity, we further obtain
\begin{equation}
\sum_{t'=1}^{t} \sum_{k \in \mathcal{K}(t') } R_k(\mathcal{K}(t') ) \geq \sum_{t'=1}^{t-1} \sum_{k \in \mathcal{K}(t') } R_k(\mathcal{K}(t')),
\end{equation}
which manifests the fact that the objective function of problem~\eqref{Problemv1} is non-decreasing along with iterations. For a given set of transmit power coefficients, the instantaneous throughput of scheduled user is finite. Hence, the objective function of problem~\eqref{Problemv1} is upper bounded and Algorithm~\ref{Algorithm1} converges to a fixed point solution. Additionally, \eqref{eq:Ser2} ensures the QoS requirements and therefore we conclude the proof.
\end{proof}
After adding user~$k_m^{t,\ast}$ to the system, we should update the available- and scheduled-user sets $\mathcal{N}$ and $\mathcal{K}(t)$ as 
\begin{equation} \label{eq:Updatev1}
	\mathcal{N}(t) \leftarrow \mathcal{N}(t) \setminus \{ k_m^{t,\ast} \} \mbox{ and } \mathcal{K}_m(t) \leftarrow \widetilde{\mathcal{K}}_{m}^{\ast} (t).
\end{equation}
The inner loop will continue until $m = M$ and the scheduled-user set $\mathcal{K}(t)$ is defined as
\begin{equation} \label{eq:Kt}
\mathcal{K}(t) \leftarrow \widetilde{\mathcal{K}}_M(t).
\end{equation}
At the end of each outer iteration, the algorithm should remove scheduled users from service if they are already satisfied their QoS requirements. This is done by computing the aggregated throughput in \eqref{eq:Rk}, and checking the QoS condition:
\begin{equation} \label{eq:Conditionv1}
R_k ( \{ \mathcal{K} (t)\}) \geq \xi_k.
\end{equation}
Let us denote $\widehat{\mathcal{K}}(t) \subseteq \mathcal{K}(t)$ the set of scheduled users already satisfied their QoS requirements, $\mathcal{K}(t)$ is further updated as
\begin{equation} \label{eq:Ktv1}
	\mathcal{K}(t) \leftarrow \mathcal{K}(t)\setminus \widetilde{\mathcal{K}}(t).
\end{equation}
The iterative approach will continue until all the time slots are considered and the proposed heuristic approach is summarized in Algorithm~\ref{Algorithm1}. Despite the local user scheduling solution, our proposed approach ensures the long-term sum throughput maximization over many different time slots with respect to their individual QoS requirements. The computational complexity of Algorithm~\ref{Algorithm1} is analytically presented hereafter.
\begin{algorithm}[t]
	\caption{A user scheduling algorithm for problem~\eqref{Problemv1}} \label{Algorithm1}
	\textbf{Input}: Available-user set $\mathcal{N}(0) \leftarrow \{1, \ldots, N \}$; Scheduled-user set $\mathcal{K}(0) \leftarrow \emptyset$; Propagation channel vectors $\{\mathbf{h}_1, \ldots, \mathbf{h}_N \}$; QoS requirements $\{ \xi_1, \ldots, \xi_k \}$; Number of time slots $T$ and individual scheduled time slots $\{T_1, \ldots, T_N\}$; Transmit data powers $\{ p_1, \ldots, p_{N} \}$.
	\begin{itemize}
		\item[1.] Select scheduled user~$\pi_1$ based on the best channel gain as obtained in \eqref{eq:ChannelOrder}.
		\item[2.] Set $t=1$, then update $\mathcal{N}(1)$ and $\mathcal{K}(1)$ by \eqref{eq:1stUpdate}.
		\item[3.] \textbf{while} $t \leq T$ \textbf{do}
		\begin{itemize}
			\item[3.1.] Set $m = |\mathcal{K}(t-1)|$ and  $\mathcal{K}_m(t) = \mathcal{K}(t-1)$.
			\item[3.2.] \textbf{while} $m \leq M$ \textbf{do}
			\begin{itemize}
			\item[3.2.1.] Obtain user~$k_{m}^{t,\ast}$ and $\widetilde{\mathcal{K}}_m^{\ast}(t)$  by  solving problem \eqref{Prob:Ratev1} with $\widetilde{\mathcal{K}}_m(t)$ updated in \eqref{eq:Ktilde}.
			\item[3.2.2.] \textbf{If} the conditions \eqref{eq:Ser1} and \eqref{eq:Ser2} satisfy: Update $\mathcal{N}(t)$ and $\mathcal{K}_m(t)$ as \eqref{eq:Updatev1}. \textbf{Otherwise} keep $\mathcal{N}(t)$ and $\mathcal{K}_m(t)$ unchanged and go to Step $3.2.3.$
			\item[3.2.3.] Set $m=m+1$.
			\end{itemize}
		    \item[3.3.] \textbf{end while}
			\item[3.4.] Update $\mathcal{K}(t)$ by \eqref{eq:Kt} and compute the throughput of scheduled users by \eqref{eq:Capacityk}.
			\item[3.5.] Find the scheduled users satisfied their QoS requirements (set $\widehat{\mathcal{K}}(t)$) by computing the aggregated throughput using \eqref{eq:Rk} and checking the condition \eqref{eq:Conditionv1}, then remove them from service by using \eqref{eq:Ktv1}.
			\item[3.6.] Update $\mathcal{K}(t)$ by \eqref{eq:Ktv1} and set $t=t+1$.
		\end{itemize}
		\item[4.] \textbf{end while}
	\end{itemize}
	\textbf{Output}: The scheduled users in the observed window time and their throughput [Mbps]. 
	\vspace*{-0.0cm}
\end{algorithm}
\vspace*{-0.2cm}
\subsection{Computational Complexity}
\vspace*{-0.1cm}
Let us consider the multiplications, division, square root, and matrix inversion as the dominated arithmetic operations, similar to \cite{van2019power,Bjornson2017bo}, the computational complexity order of  Algorithm~\ref{Algorithm1} is given in Lemma~\ref{eq:Complexity}.
\vspace{-0.2cm}
\begin{lemma} \label{eq:Complexity}
Algorithm~\ref{Algorithm1} has the computational complexity in the order of $\mathcal{O}\left( C_0 + C_1 + C_2 \right)$, where
\begin{align}
C_0 &= NM + N\log_2 N,\\
C_1 &= (M+2) \sum_{t=1}^T \sum_{m= |\mathcal{K}(t-1)|}^M  |\mathcal{N}(t)| |\widetilde{\mathcal{K}}_m(t)|,\\
C_2 & = \frac{M^2}{2}\sum_{t=1}^T \sum_{m=|\mathcal{K}(t-1)|}^M |\mathcal{N}(t)|   |\widetilde{\mathcal{K}}_m(t)|^2.
\end{align}
\end{lemma}
\begin{proof}
Selecting the first scheduled user based on the channel gains requires the $N(M+1)$ arithmetic operations to compute the $N$ channel gains and $\mathcal{O}(N\log_2 N)$ for sorting them in a descending order as in \eqref{eq:ChannelOrder}. Therefore, the computational complexity of this step is proportional to $N( M+ 1 + \log_2 N )$. For each inner loop, we first need to compute the instantaneous throughput in \eqref{eq:Capacityk}, which requires the $(M+2)|\widetilde{\mathcal{K}}_m(t)| +3$ arithmetic operations. The computational complexity needed to solve \eqref{Prob:Ratev1} scales up with the factor $ |\mathcal{N}(t)|(M+2) |\widetilde{\mathcal{K}}_m(t)|+ 3 |\mathcal{N}(t)| $, thus the inner loop has the computational complexity in the order of $ |\mathcal{N}(t)|(M+2) \sum_{m= |\mathcal{K}(t-1)|}^M  |\widetilde{\mathcal{K}}_m(t)|$. Furthermore, each RZF precoding matrix with the cost as in Lemma~\ref{lemma:ComplexityWtv2} leads to the total computational complexity per inner loop in the order of $ \frac{1}{2}|\mathcal{N}(t)|M^2 \sum_{m=|\mathcal{K}(t-1)|}^M  |\widetilde{\mathcal{K}}_m(t)|^2$. By summing up all the cost and removing the terms with low degree, the result is obtained as in the lemma.
\vspace{-0.2cm}
\end{proof}
Lemma~\ref{eq:Complexity} manifests that Algorithm~\ref{Algorithm1} has the computational complexity per iteration in a quadratic order of the scheduled users and satellite beams, whereby the entire computational complexity is much lower than an exhaustive search. This algorithm can thus perform the user scheduling for a large-scale network with many users.
\begin{figure*}[t]
	\begin{minipage}{0.33\textwidth}
		\centering
		\includegraphics[trim=3.4cm 10.2cm 4.0cm 10.6cm, clip=true, width=2.35in]{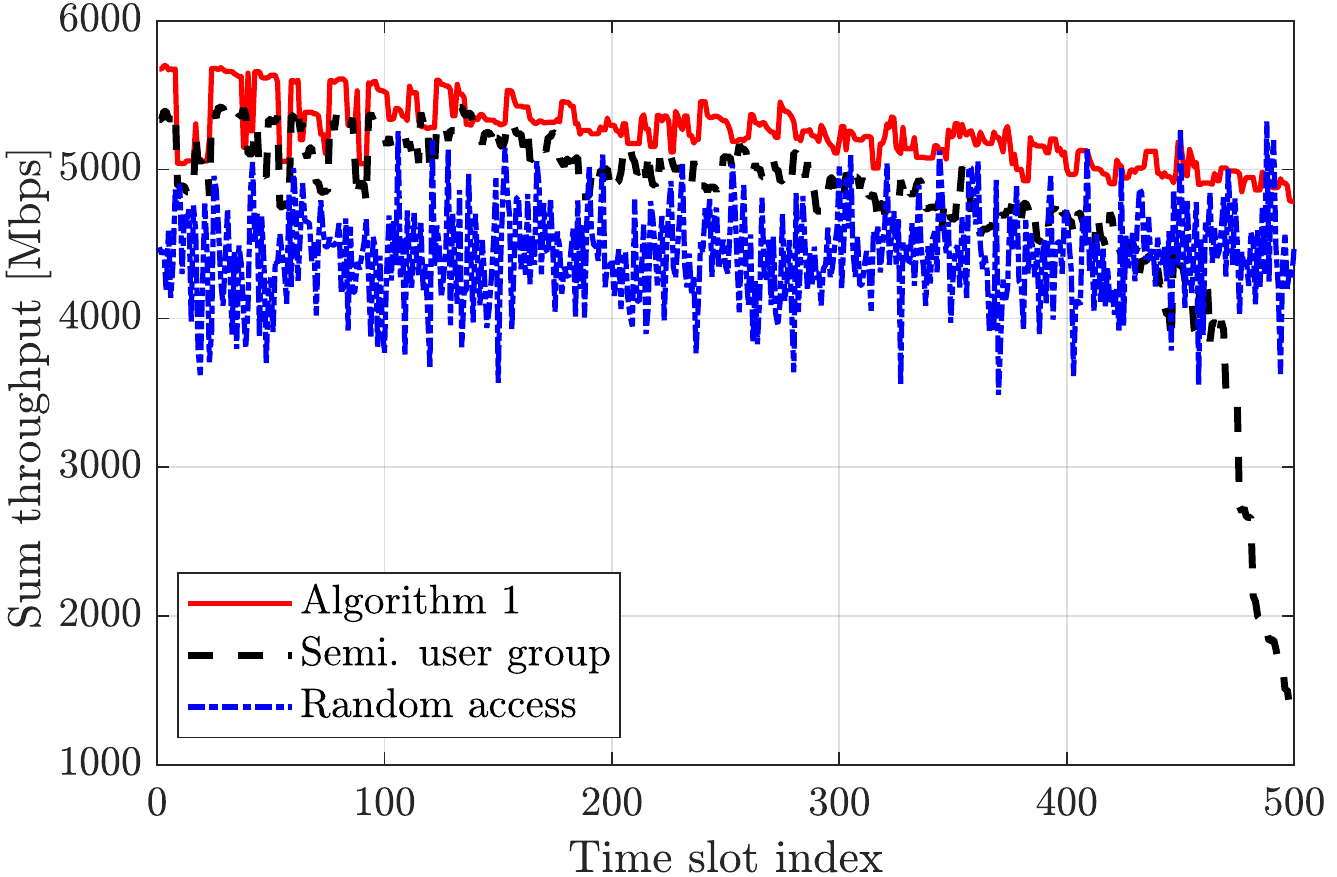} \vspace*{-0.15cm}\\
		\fontsize{8}{8}{(a)}
		\vspace*{-0.2cm}
	\end{minipage}
	\begin{minipage}{0.33\textwidth}
		\centering
		\includegraphics[trim=3.4cm 10.2cm 4.0cm 10.6cm, clip=true, width=2.35in]{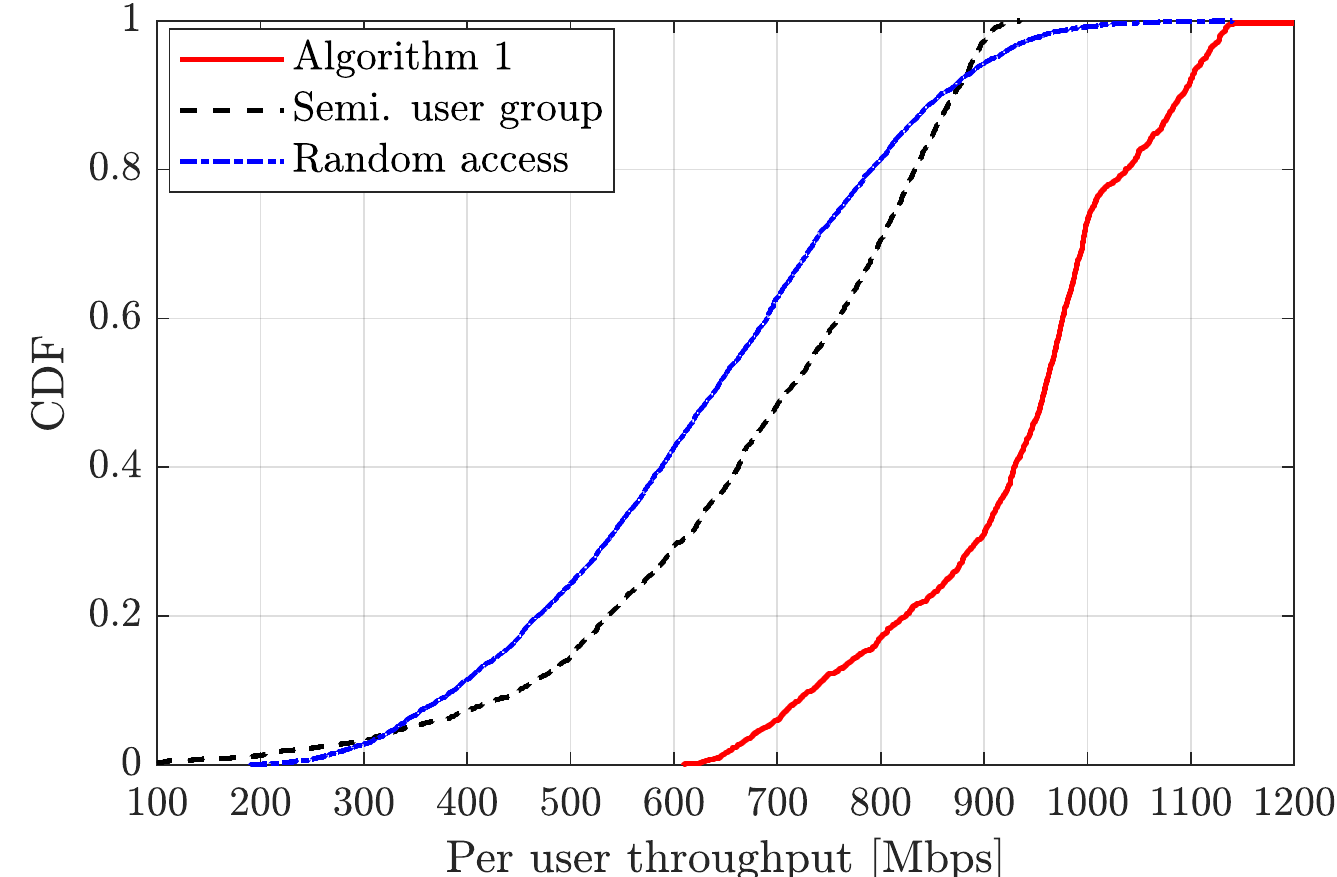} \vspace*{-0.15cm}\\
		\fontsize{8}{8}{(b)}
		\vspace*{-0.2cm}
	\end{minipage}
	\begin{minipage}{0.33\textwidth}
		\centering
		\includegraphics[trim=3.4cm 10.2cm 4.0cm 10.6cm, clip=true, width=2.35in]{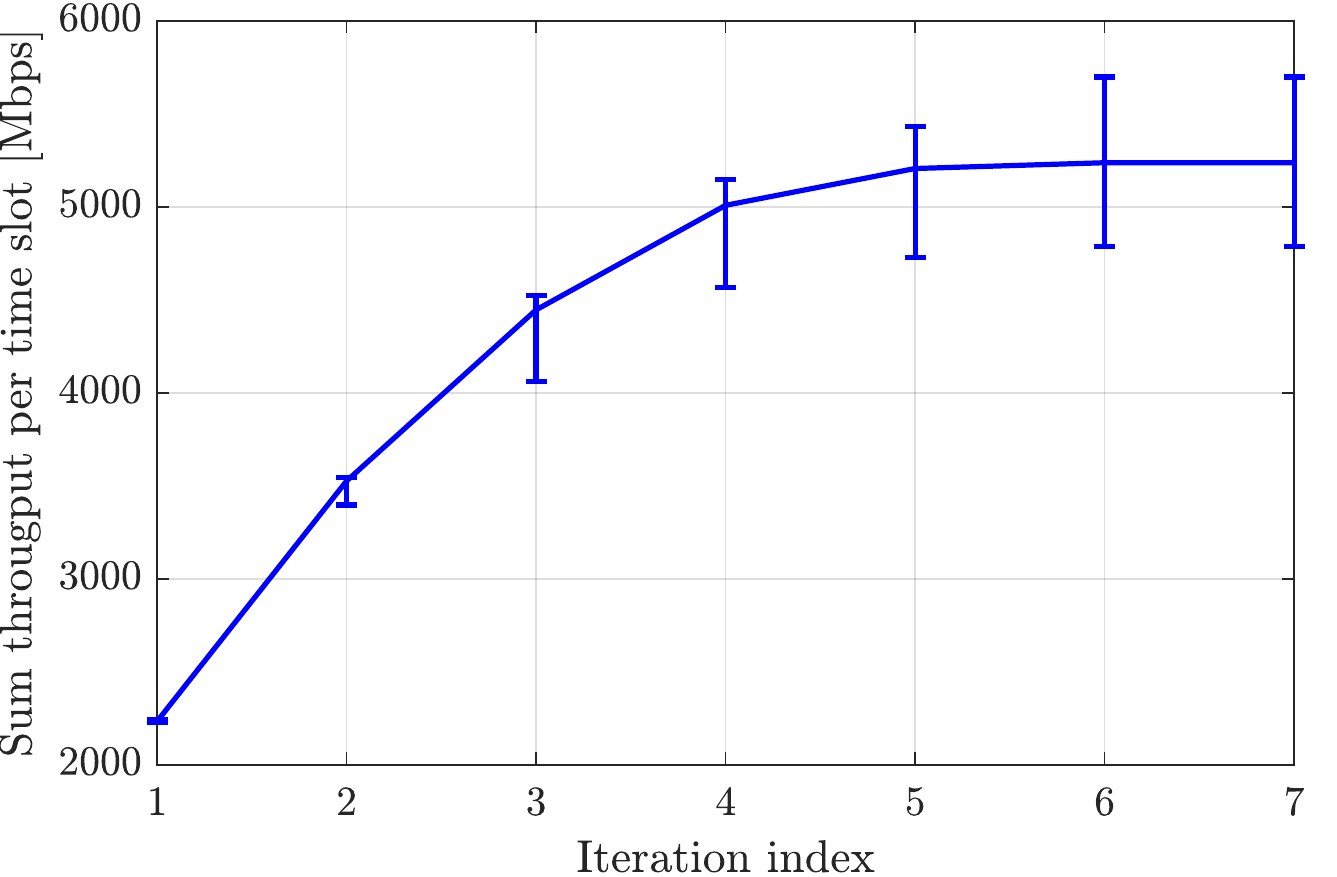} \vspace*{-0.15cm}\\
		\fontsize{8}{8}{(c)}
		\vspace*{-0.2cm}
	\end{minipage}
	\caption{The system performance: (a) The sum throughput [Mbps] versus the time slots; (b) The throughput [Mbps] per user; (c) The convergence property of the sum throughput [Mbps] per time slot.} \label{Fig2}
	\vspace*{-0.4cm}
\end{figure*}
\vspace*{-0.2cm}
\section{Numerical Results}
\vspace*{-0.1cm}
We consider a GEO satellite system $3500$ users clustered into $7$ beams. The observed window time includes $500$ time slots. A sum power-constrained system is considered with the per-beam power of $10$~dBW. The system bandwidth is $500$ MHz and the carrier frequency is $19.95$~GHz. The QoS requirement per time slot is $500$~Mbps with the total time slots per user $T_k = 1, \forall k,$ for simplicity. In order to demonstrate the efficiency of the proposed optimization framework, the following benchmarks are included for comparison:
\vspace{-0.1cm}
\begin{itemize}
\item[$i)$] \textit{Proposed heuristic algorithm} is presented in Algorithm~\ref{Algorithm1} via working on the user scheduling to maximize the sum throughput with the QoS requirements. 
\item[$ii)$] \textit{Semiorthogonal user group} was proposed in \cite{Yoo2006a} by exploiting the orthogonality between the propagation channels. The number of scheduled users and satellite beams are assumed to be equal. Additionally, the user scheduling does not include the QoS requirements into account.
\item[$iii)$] \textit{Random access} is a low computational complexity benchmark and served as the baseline in previous works \cite{8371220}. Along with time slots, the number of scheduled users are randomly selected and equal to the number of satellite beams. There is no guarantee on the QoS requirements.
\vspace{-0.00cm}
\end{itemize}
Figure~\ref{Fig2}(a) plots the sum throughput [Mbps] as a function of time slots. Random access provides the worst throughput in most of the time slots that is only $4419$ [Mbps] on average. However, it offers good performance in the last time slots. Semiorthogonal user group performs $7.6\%$ sum throughput better than random access with $4753$ [Mbps] on average and becomes the worst in the last time slots where the available users have strongly correlated channels. Algorithm~\ref{Algorithm1} gives the best performance with $18.6\%$ better than the baseline. 

Figure~\ref{Fig2}(b) shows the cumulative density function (CDF) of the scheduled users. Random access averagely provides the throughput of about $631$ [Mbps] per user, while semiorthogonal user group offers $679$ [Mbps]. Notably, Algorithm~\ref{Algorithm1} gives the highest per-user throughput with $1.37\times$ higher than semiorthogonal user group. Algorithm~\ref{Algorithm1} ensures all the scheduled users with their QoS requirements. In contrast,  $24.4\%$ and $14.5\%$ user locations cannot be served with the requested QoS if the system deploys random access and semiorthogonal user group, respectively, due to no QoS guarantee in those benchmarks. It manifests the practical importance of Algorithm~\ref{Algorithm1}.

Figure~\ref{Fig2}(c) plots the convergence of Algorithm~\ref{Algorithm1} by utilizing the median rate among the $500$ time slots. Significant growth of the sum throughput is observed in the first iterations, then reaching the fixed point when the iteration index equals the number of satellite beams. The sum throughput at the last iteration improves $2.3\times$ compared to the first one. Furthermore, the error bars show the fluctuation at each time slot compared to the median value. From a small fluctuation at the beginning, it gets larger in the last iterations. Consequently, a good scheduling plays a critical role in improving  the sum throughput while maintaining the QoSs. 
\vspace*{-0.1cm}
\section{Conclusion}
\vspace*{-0.1cm}
This paper proposes a heuristic user scheduling strategy for large-scale MB-HTS systems where many users simultaneously request to access the network. We formulated a total throughput optimization maximization problem in an observed window time subject to the individual QoS requirements. Due to the inherent non-convexity, we proposed a heuristic algorithm to obtain a local solution with low computational complexity. Numerical results demonstrated all scheduled users having the better QoSs than requested. Besides, the proposed algorithm offers  better sum throughput [Mbps] per time slot than the other benchmarks with up to $18.6\%$.
\vspace*{-0.2cm}
\bibliographystyle{IEEEtran} 
\bibliography{IEEEabrv,refs}
\end{document}